%% file: roundcompressfull.tex
\documentclass[11pt, letterpaper]{article}
\usepackage{amsmath, pdfsync}
\usepackage[usenames]{color}
\usepackage[full,nousetoc, nouselot,titlepage,final,nofullpage,nohylinks]{boaz3}
\input{macros}

\usepackage{fullpage}

\title{Information Equals Amortized Communication}
\author{ Mark Braverman\thanks{University of Toronto, \texttt{mbraverm@cs.toronto.edu}. Supported by an NSERC Discovery Grant.} \and  Anup Rao\thanks{University of Washington, {\tt anuprao@cs.washington.edu}. Supported by the National Science Foundation under agreement CCF-1016565.}}

\renewcommand{\e}{\varepsilon}

\newcommand{\AProt}[1]{\mathsf{AC}(#1)}
\newcommand{\CProt}[1]{\mathsf{CC}(#1)}
\newcommand{\CFunc}[1]{\mathsf{R}(#1)}
\newcommand{\IProt}[2]{\mathsf{IC^i}_{#2}(#1)}

\renewcommand{\mathbb}[1]{{\bf #1}}
\newcommand{\childa}[1]{\mathsf{child^A_{#1}}}
\newcommand{\childb}[1]{\mathsf{child^B_{#1}}}
\newcommand{\child}{\mathsf{child}}
\newcommand{\DivProt}[1]{\mathbb{D} \left (  #1 \right) }
\newcommand{\Div}[2]{\mathbb{D} \left (  #1 ||#2 \right) }

\newcommand{\eqdef}{\stackrel{def}{=}}

\newcommand{\CPJ}{\mathsf{CPJ}}

\newcommand{\papprox}[1]{\stackrel{#1}{\approx}}

\ifnum\hylinks=1
\newcommand{\namedref}[2]{\hyperref[#2]{#1~\ref*{#2}}}
\else
\newcommand{\namedref}[2]{#1~\ref{#2}}
\fi

\newcommand{\theoremref}[1]{\namedref{Theorem}{#1}}

\newcommand{\figureref}[1]{\namedref{Figure}{#1}}

\newcommand{\lemmaref}[1]{\namedref{Lemma}{#1}}

\newcommand{\propositionref}[1]{\namedref{Proposition}{#1}}

\newcommand{\corollaryref}[1]{\namedref{Corollary}{#1}}
\newcommand{\equationref}[1]{\namedref{Equation}{#1}}

\renewcommand{\P}{{\mathbf P}}
\newcommand{\cA}{{\cal A}}
\newcommand{\cU}{{\cal U}}

\newcommand{\cP}{{\cal P}}
\newcommand{\cQ}{{\cal Q}}
\ifnum\full=1

\else

\fi

\newcommand{\Ex}[2]{\mathop{\mathbb{E}}\displaylimits_{#1}\left
[ #2 \right ]}
\newcommand{\Expect}[1]{\mathop{\mathbb{E}}\left
[ #1 \right ]}

\newcommand{\ve}{\varepsilon}
\newcommand{\al}{\alpha}
\newcommand{\algwidth}{0.97\textwidth}
\begin{document}

\begin{DOCheader}
\begin{abstract}

We show how to efficiently simulate the sending of a message $M$ to a receiver who has partial information about the message, so that the expected number of bits communicated in the simulation is close to the amount of additional information that the message reveals to the receiver. This is a generalization and strengthening of the Slepian-Wolf theorem, which shows how to carry out such a simulation with low \emph{amortized} communication in the case that $M$ is a deterministic function of $X$. A caveat is that our simulation is interactive. 

As a consequence, we prove that the internal information cost (namely the information revealed to the parties) involved in computing any relation or function using a two party interactive protocol is {\em exactly} equal to the amortized communication complexity of computing independent copies of the same relation or function. We also show that the only way to prove a strong direct sum theorem for randomized communication complexity is by solving a particular variant of the pointer jumping problem that we define. Our work implies that a strong direct sum theorem for communication complexity holds if and only if efficient compression of communication protocols is possible.

% 
% Here by amortized communication complexity we mean the average per copy communication in the best protocol for computing multiple copies, with a bound on the error in each copy (i.e.\ we require only that the output in each coordinate is correct with good probability, and do not require that all outputs are simultaneously correct). This significantly simplifies the relationships between the various measures of complexity for average case communication protocols, and proves that if a function's information cost is smaller than its communication complexity, then multiple copies of the function can be computed more efficiently in parallel than sequentially.

% If this problem has a cheap communication protocol, then a strong direct sum theorem must hold. On the other hand, if it does not, then the problem itself gives a counterexample for the direct sum question. In the process we show that a strong direct sum theorem for communication complexity holds if and only if efficient compression of communication protocols is possible. 

\end{abstract}
\end{DOCheader}

	\section{Introduction}
	Suppose a sender wants to transmit a message $M$ that is correlated with an input $X$ to a receiver that has some information $Y$ about $X$. What is the best way to carry out the communication in order to minimize the expected number of bits transmitted? A natural lower bound for this problem is the mutual information between the message and $X$, given $Y$: $I(M;X|Y)$, i.e. the amount of 
	new information $M$ reveals to the receiver about $X$. In this work, we give an interactive protocol that has the same effect as sending $M$, yet the expected number of bits communicated is asymptotically close to optimal --- it is the same as the amount of new information that the receiving party learns from $M$, up to a sublinear additive term\footnote{Observe that if $X,Y,M$ are arbitrary random variables, and the two parties are tasked with sampling $M$ efficiently (as opposed to one party transmitting and the other receiving), it is impossible to succeed in communication comparable to the information revealed by $M$. For example, if $M = f(X,Y)$, where $f$ is a boolean function with high communication complexity on average for $X,Y$, $M$ reveals only one bit of information about the inputs, yet cannot be cheaply sampled.}.

	Our result is a generalization of classical data compression, where $Y$ is empty (or constant), and $M$ is a deterministic function of $X$. In this case, the information learnt by the receiver is equal to the entropy $H(M)$, and the compression result above corresponds to classical results on data compression first considered by Shannon  \cite{Shannon48} --- $M$ can be encoded so that the expected number of bits required to transmit $M$ is $H(M)+1$ (see for example the text \cite{CoverT91}).
	
	Typical work in information theory usually focuses on the easier problem of communicating $n$ independent copies $M_1,\dotsc,M_n$, where each $M_i$ has an associated dependent $X_i,Y_i$. Here $n$ is viewed as a growing parameter, and the average communication is measured. Indeed, any solution simulating a single message can be applied to simulate the transmission of $n$ messages, but there is no clear way to use an asymptotically good solution to compress a single message. By the asymptotic equipartition property of the entropy function, taking independent copies essentially forces most of the probability mass of the distributions to be concentrated on sets of the ``right'' size, which simplifies this kind of  problem significantly. The Slepian-Wolf theorem \cite{SlepianW73} addresses the case when $M$ is determined by $X$. The theorem states that there is a way to encode many independent copies $M_1,\dotsc,M_n$ using roughly $I(M;X|Y)$ on average, as $n$ tends to infinity. The theorem and its proof do not immediately give any result for communicating a single message. Other work has focused on the problem of generating two correlated random variables with minimal communication \cite{Cuff08}, and understanding the minimal amount of information needed to break the dependence between $X,Y$ \cite{Wyner75}, neither of which seem useful to the problem we are interested in here.
	
Motivated by questions in computer science, prior works have considered the problem of encoding a single message where $M$ is not necessarily determined by $X$ (see \cite{JainRS03, HarshaJMR07} and the references there), but these works do not handle the case above, where the receiver has some partial information about the sender's message.

\section{Consequences in Communication Complexity}

Given a function $f(x,y)$, and a distribution $\mu$ on inputs to $f$, there are several ways to measure the complexity of a communication protocol that computes $f$.
\begin{itemize}
 \item The communication complexity $D^\mu_\rho$, namely the maximum number of bits communicated by a protocol that computes $f$ correctly except with probability $\rho$.
\item The amortized communication complexity, $\lim_{n \to \infty} D^{\mu,n}_\rho/n$, where here $D^{\mu,n}_\rho$ denotes the communication involved in the best protocol that computes $f$ on $n$ independent pairs of inputs drawn from $\mu$, getting the answer correct except with probability $\rho$ in each coordinate.
\end{itemize}

 Let $\pi(X,Y)$ denote the public randomness and messages exchanged when the protocol $\pi$ is run with inputs $X,Y$ drawn from $\mu$. Another set of measures arises when one considers exactly how much information is revealed by a protocol that computes $f$.
\begin{itemize}
	\item The minimum amount of information that must be learnt about the inputs by an observer who watches an execution of any protocol ($I(XY ; \pi(X,Y))$) that compute $f$ except with probability of failure $\rho$, called the \emph{external information cost} in \cite{BarakBCR10}.
	\item The minimum amount of new information that the parties learn about each others input by executing any protocol ($I(X;\pi(X,Y)|y)+ I(Y;\pi(X,Y)|X)$) that computes $f$ except with probability of failure $\rho$, called the \emph{internal information cost} in \cite{BarakBCR10}. In this paper we denote this quantity $\IProt{f,\rho}{\mu}$.
	\item The amortized versions of the above measures, namely the average external/internal information cost of a protocol that computes $f$ on $n$ independent inputs correctly except with probability $\rho$ in each coordinate.
\end{itemize}

Determining the exact relationship between the amortized communication complexity and the communication complexity of the function is usually referred to as the \emph{direct sum} problem, which has been the focus of much work \cite{ChakrabartiSWY01, Shaltiel03, JainRS03, HarshaJMR07, BarakBCR10, Klauck10}. For randomized and average case complexity, we know that $n$ copies must take approximately (at least) $\sqrt{n}$ times the communication of one copy, as shown by the authors with Barak and Chen \cite{BarakBCR10}. For worst case (deterministic) communication complexity, Feder, Kushilevitz, Naor, and Nisan~\cite{FederKNN95} showed that if a single copy of a function $f$ requires $C$ bits of
communication, then $n$ copies require $\Omega(\sqrt{C}n)$ bits. In the rest of the discussion in this paper, we focus on the average case and randomized communication complexity. 

The proofs of the results above for randomized communication complexity have a lot to do with the information theory based measures for the complexity of communication protocols. Chakrabarti, Shi, Wirth and Yao \cite{ChakrabartiSWY01} were the first to define the external information cost, and prove that if the inputs are independent in $\mu$, then the external information cost of $f$ is at most the amortized communication complexity of $f$. This sparked an effort to relate the amortized communication complexity to the communication complexity. If one could compress any protocol so that the communication in it is bounded by the external information cost, then, at least for product distributions $\mu$, one would show that the two measures of communication complexity are the same. 

For the case of general distributions $\mu$, it was shown in \cite{BaryossefJKS04,BarakBCR10} that the amortized communication complexity can only be larger than the internal information cost. In fact, the internal and external information costs are the same when $\mu$ is a product distribution, so the internal information cost appears to be the appropriate measure for this purpose. \cite{BarakBCR10} gave a way to compress protocols so that the communication is reduced to the geometric mean of the internal information and the communication in the protocol, which gave the direct sum result discussed above. % Although it is easy to bound the (internal/external) information cost in terms of the communication in the protocol, proving lower bounds on the information cost is much more challenging. One of the consequences of \cite{BarakBCR10} is that the information cost of the inner product function under the uniform distribution must be $\Omega(n/\log n)$ (since otherwise one could compress the protocol to get an impossibly cheap communication protocol for inner product).

	The main challenge that remains is to find a more efficient way to compress protocols whose internal information cost is small. Indeed, as we discuss below, in this paper we show that this is essentially the \emph{only} way to make progress on the direct sum question, in the sense that if  some protocol  cannot be compressed well, then it can be used to define a function whose amortized communication complexity is significantly smaller than its communication complexity.
	
	\subsection{Our Results}

	\begin{figure}[ht] 
	\begin{center} \includegraphics[angle=0,scale=.8]{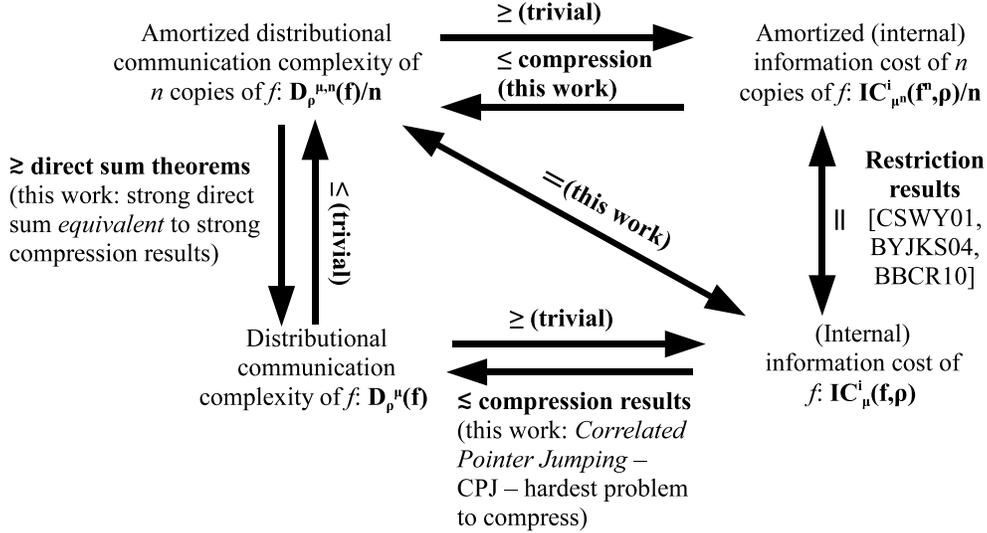}
	\caption{The relationships between different measures of complexity for communication problems, with the new results highlighted. 
	The present work collapses the upper-right triangle in the diagram, showing that amortized communication complexity is equal to the 
	internal information cost of any functionality. We further show three equivalent characterizations that would lead to a collapse in 
	the lower-left triangle: strong direct sum theorems, near-optimal protocol compression and solving the Correlated Pointer Jumping 
	efficiently.}
	\label{fig:roadmap}
	\end{center}
	\end{figure}

	Our main technical result is a way to compress one round protocols according to the internal information cost, which corresponds to the problem of efficiently communicating information when the receiver has some partial information, discussed in the introduction. In fact, we design a protocol that solves a harder problem, that we describe next. We give a way for two parties to efficiently sample from a distribution $P$ that only   one of them knows, by taking advantage of a distribution $Q$ known only to the other. We obtain a protocol whose communication complexity can be bounded in terms of the informational divergence $\Div{P}{Q} = \sum_{x} P(x) \log(P(x)/Q(x))$.

	\begin{theorem}\label{theorem:main}
	Suppose that player $A$ is given a distribution (described by the probabilities assigned to each point) $P$ and player $B$ is given a distribution $Q$ over a universe $\cU$. 
	There is a public coin protocol that uses an expected $$\Div{P}{Q}+2\log(1/\e)+O\left (\sqrt{\Div{P}{Q}}+1 \right )$$ bits of communication such that 
	at the end of the protocol:
	\begin{itemize}
	\item 
	Player $A$ outputs an element $a$ distributed according to $P$;
	\item 
	Player $B$ outputs $b$ such that for each $x \in \cU$, $\P[b=x|~a=x]>1-\e$. 
	\end{itemize}
	\end{theorem}

	As a corollary, we obtain the formulation discussed earlier. For any distribution $X,Y$ and message $M$ that is independent of $Y$ once $X$ is fixed, we can have the sender set $P$ to be the distribution of $M$ conditioned on her input $x$, and the receiver set $Q$ to be the distribution of $M$ conditioned on her input $y$. The expected divergence $\Div{P}{Q}$ turns out to be equal to the mutual information $I(M;X|Y)$. Indeed, applying \theoremref{theorem:main} to each round of communication in a multiround protocol, gives the following corollary, where setting $r=1$ gives the formulation discussed in the introduction.	The proof appears in Section \ref{subsec:compress}. 
	
	% 
	% \begin{corollary} For any random variables $X,Y,M$ with the property that $Y$ is independent of $M$ given any fixing of $X$, and every $\e >0$ there is a protocol for two parties to sample $M'$ given inputs $X,Y$ such that the expected communication of the protocol is at most $$I(M;X|Y) + \log(1/\e) + O(\sqrt{I(M;X|Y)} + 1), $$ and the resulting distribution $X,Y,M'$ is $\e$-close to the distribution of $X,Y,M$.
	% \end{corollary}
	% 
	% Applying this corollary to every round of a communication protocol gives a compression scheme.

	\begin{corollary}\label{corollary:compress}
		Let $X,Y$ be inputs to a $k$ round communication protocol $\pi$ whose internal information cost is $I$. Then for every $\e>0$, there exists a protocol $\tau$ such that at the end of the protocol, each party outputs a transcript for $\pi$. Furthermore, there is an event $G$ with $\P[G]>1-k\ve$ such that conditioned on $G$, the expected communication of $\tau$ is $I + O(\sqrt{kI}+k) + 2k \log(1/\ve)$, and both parties output the same transcript distributed exactly according to $\pi(X,Y)$.
	\end{corollary}

	This compression scheme significantly clarifies the relationship between the various measures of complexity discussed in the introduction. In particular, it allows us to prove that the internal information cost of computing a function $f$ according to a fixed distribution is {\em exactly} equal to the amortized communication complexity of computing many copies of $f$. 
	\begin{theorem} 
	For any $f$, $\mu$, and $\rho$, 
	$$
	\IProt{f,\rho}{\mu} = \lim_{n\rightarrow \infty} \frac{D^{\mu,n}_\rho(f)}{n}.
	$$
	\end{theorem}
	
	This result seems surprising to us, since it characterizes the information cost in terms of a quantity that at first seems to have no direct connection to information theory. The proof appears in Appendix \ref{subsec:infoequals}. It proves that if a function's information cost is smaller than its communication complexity, then multiple copies of the function can be computed more efficiently in parallel than sequentially. Observe that the naive sequential protocol for computing multiple copies would only give a bound on the error in each copy separately (exactly as in our definition of amortized communication complexity). The consequences to the various measures discussed earlier are summarized in Figure \ref{fig:roadmap}.

In Appendix \ref{sec:53}, we define a communication problem we call Correlated Pointer Jumping -- $\CPJ(C,I)$ -- that is parametrized by two parameters $C$ and $I$ 
such that $C\gg I$. $\CPJ(C,I)$ is designed in a way that the randomized communication complexity cost $I\le R(\CPJ(C,I))\le C$. We show that determining the worst case randomized communication complexity $R(\CPJ(C,I))$ for $I=C/n$ is equivalent (up to poly-logarithmic factors) to determining the best parameter $k(n)$ for which a direct sum theorem $R(f^n)=\Omega(k(n) \cdot R(f))$ holds. For simplicity, we 
 formulate only part of the result here.

\begin{theorem}
If $\CFunc{\CPJ(C,C/n)}=\tilde{O}(C/n)$ for all $C$, then a near optimal direct sum theorem holds: $\CFunc{f^n}=\tilde{\Omega}(n \cdot \CFunc{f})$
for {\em all} $f$. 

On the other hand, if $\CFunc{\CPJ(C,C/n)}= \Omega((C \log^a C)/n)$ for all $a>0$, then direct sum is violated by $\CPJ(C,C/n)$:
$$
\CFunc{\CPJ(C,C/n)^n} = O(C\log C) = o(n \cdot \CFunc{\CPJ(C,C/n)}/\log^a C),
$$
for all $a$. 
\end{theorem}

Finally, letting $f^n$ denote the function that computes $n$ copies of $f$ on $n$ different inputs, our protocol compression yields the following direct sum theorem:

\begin{corollary} [Direct Sum for Bounded Rounds]
Let $C$ be the communication complexity of the best protocol for computing $f$ with error $\rho$ on inputs drawn from $\mu$. Then any $r$ round protocol computing $f^n$ on the distribution $\mu^n$ with error $\rho - \ve$ must involve at least $\Omega(n(C-r \log(1/\ve)-O(\sqrt{C\cdot r})))$ communication. 
\end{corollary}

	\subsection{Techniques}
	The key technical contribution of our work is a sampling protocol that proves \theoremref{theorem:main}. The sampling method we show is different from the ``Correlated Sampling'' technique used in work on parallel repetition \cite{Holenstein07, Rao08a} and in the previous paper on compression \cite{BarakBCR10}. In those contexts it was guaranteed that the input distributions $P,Q$ are close in \emph{statistical distance}. In this case, the sampling can be done without any communication. In our case, all interesting inputs $P,Q$ are very far from each other in statistical distance, and not only that, but the ratios of the probabilities $P(x)/Q(x)$ may vary greatly with the choice of $x$. It is impossible to solve this problem without communication, and we believe it is unlikely that it can be solved without interaction.
	
	Indeed, our sampling method involves interaction between the parties, and for good reasons. In the case that the sample is $x$ for which $P(x)/Q(x)$ is very large, one would expect that a lot of communication is needed to sample $x$, since the second party would be surprised with this sample, while if $P(x)/Q(x)$ is small, then one would expect that a small amount of communication is sufficient. Our protocol operates in rounds, gradually increasing the number of bits that are communicated until the sample is correctly determined. 

To illustrate our construction, consider the baby case of the problem where the issue of high variance in $P(x)/Q(x)$ does not affect us. Recall that the informational divergence $\Div{P}{Q}$ is equal to $\sum_{x} P(x) \log \frac{P(x)}{Q(x)}$. Suppose $Q$ is the uniform distribution on some subset $S_Q$ of the universe $\cU$, and $P$ is the uniform distribution on some subset $S_P \subset S_Q$. Then the informational divergence $\Div{P}{Q}$ is exactly $\log(|S_Q|/|S_P|)$.

	In this case, the players use an infinite public random tape that samples an infinite sequence of elements $a_1,a_2,\dotsc$ uniformly at random from the universe $\cU$. Player $A$ then picks the first element $x$ that lies in $S_P$ to be his sample.   Next the players use the public randomness to sample a sequence of uniformly random boolean functions on the universe. $A$ then sends a stream of these functions evaluated at $x$. At each round $i$, player $B$ finds the first element $y_i$ on the tape that belongs to $S_Q$ and is consistent with the values Player $A$ has sent so far. Player $B$ uses $y_i$ as his working hypothesis for the element Player $A$ is trying to communicate.
Player $B$ lets Player $A$ know (and outputs $y_i$) if the element $y_i$ stays the same for some interval $i=[j..2 j + \log 1/\ve]$. That is, when the hypothesis for the element $x$ stops changing. For the analysis, one 
has to note that the (expected) number of elements that $B$ will have to reject before converging to $x$ is bounded in terms of $\log |S_Q|/|S_P|$ -- the divergence between $P$ and $Q$ in this case. 
	
%	If this element is $a_i$, player $A$ sends $k = \ceil{i/|\cU|}$ to player $B$. In expectation $k$ is only a constant, so the expected number of bits for this step is only a constant. Next the players use the public randomness to %sample a sequence of uniformly random boolean functions on the universe. $A$ then sends the value of approximately $\log(1/\ve)$ of these functions evaluated on his sample. $B$ looks at her window of $|\cU|$ elements and checks to %see whether any of them agree with the evaluations sent by $A$ and are in her set $S_Q$. If more than one agrees with $A$ she asks $A$ to send more evaluations of random functions. They continue this process until there is a %unique element in the $k$'th interval that agrees with the evaluations and is in the set $S_Q$. For the analysis, note that the fraction of elements in the window that are in $S_Q$ but not in $S_P$ can be bounded in terms of the %divergence between $P$ and $Q$. The general case is a little more involved, since $P(x)/Q(x)$ may vary with $x$. 	

\section{Preliminaries} \label{sec:prelims}

\paragraph{Notation.} We reserve capital letters for random variables and distributions, calligraphic letters for sets, and small letters
for elements of sets. Throughout this paper, we often use the notation $|b$ to denote conditioning on the event $B=b$.
Thus $A|b$ is shorthand for $A|B=b$. 

We use the standard notion of \emph{statistical}/\emph{total variation} distance between two distributions.

\begin{definition}
Let $D$ and $F$ be two random variables taking values in a set $\mathcal{S}$. Their \emph{statistical distance} is
    \begin{align*}
        | D-F | \eqdef \max_{\mathcal{T} \subseteq \mathcal{S}}(|\Pr[D \in \mathcal{T}] - \Pr[F \in \mathcal{T}]|) = \frac{1}{2} \sum_{s \in \mathcal{S}}|\Pr[D=s]-\Pr[F=s]|
    \end{align*}
    If $|D-F| \leq \e$ we shall say that $D$ is \emph{$\e$-close} to $F$. We shall also use the notation $D
\papprox{\e} F$ to mean $D$ is $\e$-close to $F$.
\end{definition}

\subsection{Information Theory}

\begin{definition}[Entropy] The \emph{entropy} of a random variable $X$ is $H(X) \eqdef \sum_x \Pr[X=x] \log(1/\Pr[X=x])$. The \emph{conditional entropy} $H(X|Y)$ is defined to be $\Ex{y \getsr Y}{ H(X|Y=y)}$.
\end{definition}

\begin{fact} $H(AB) = H(A) + H(B|A)$.
\end{fact}

\begin{definition}[Mutual Information]
The \emph{mutual information} between two random variables $A,B$, denoted $I(A;B)$ is defined to be the quantity $H(A)
- H(A|B) = H(B) - H(B|A)$. The \emph{conditional mutual information} $I(A;B |C)$ is $H(A|C) - H(A|BC)$.
\end{definition}

In analogy with the fact that $H(AB) = H(A) + H(B|A)$,

\begin{proposition}[Chain Rule] \label{prop:infobreak}
Let $C_1,C_2,D,B$ be random variables. Then \[I(C_1 C_2; B |D) = I(C_1;B|D) + I(C_2;B|C_1D).\]
\end{proposition}
% 
% The previous proposition immediately implies the following:
% 
% \begin{proposition} [Super-Additivity of Mutual Information] \label{prop:infoadds}
% Let $C_1,C_2,D,B$ be random variables such that for every fixing of $D$, $C_1$ and $C_2$ are independent. Then \[I(C_1; B |D)
% + I(C_2; B|D) \quad \leq \quad I(C_1 C_2;  B | D).\]
% \end{proposition}

We also use the notion of \emph{divergence} (also known as Kullback-Leibler distance or relative entropy), which is a different way to measure the distance between two
distributions:
\begin{definition}[Divergence]
The informational divergence between two distributions is $\Div{A}{B} \eqdef \sum_x A(x) \log(A(x)/B(x))$.
\end{definition}

For example, if $B$ is the uniform distribution on $\bits^n$  then $\Div{A}{B} = n - H(A)$.
% 
% \begin{proposition} \label{prop:divstat} $\Div{A}{B} \geq |A-B|^2$.
% \end{proposition}

\begin{proposition} \label{prop:divinfo}
Let $A,B,C$ be random variables in the same probability space. For every $a$ in the support of $A$ and $c$ in the
support of $C$, let $B_a$ denote $B|A=a$ and $B_{ac}$ denote $B|A=a,C=c$. Then $I(A;B |C) = \Ex{a,c \getsr
A,C}{\Div{B_{ac}}{B_{c}}}$
\end{proposition}

\begin{lemma} \label{lemma:divproduct}
$$
\Div{P_1\times P_2}{Q_1\times Q_2} = 
\Div{P_1}{Q_1}+\Div{P_2}{Q_2}.
$$
\end{lemma}
% 
% The above facts imply the following easy proposition:
% 
% {\bf [DO WE STILL NEED THIS??]}
% 
% \begin{proposition} \label{prop:infostat} With notation as in \propositionref{prop:divinfo}, for any random variables $A,B$, $$\Ex{a \getsr A}{| (B_a) - B |} \leq \sqrt{I(A;B)}.$$
% \end{proposition}
% \begin{proof}
%     \begin{align*}
%         \Ex{a \getsr A}{| (B_a) - B |} &\leq \Ex{a \getsr A}{ \sqrt{ \Div{B_a}{B}}} &\\
%         &\leq \sqrt{\Ex{a \getsr A}{  \Div{B_a}{ B}}} & \text{by convexity}\\
%         &= \sqrt{I(A;B)} & \text{by \propositionref{prop:divinfo}}
%     \end{align*}
% \end{proof}
%
% \subsection{Consistent Sampling}
%
% Let $D$ be a finite set and let $\mathcal{D}$ denote the set of distributions on this set. We have the following lemma:
%
% \begin{lemma}[Consistent Sampling] \label{lemma:consistent} There is a probability space $\Omega$ and a set of random variables, one for every distribution in $\mathcal{D}$ such that each random variable is distributed according to the corresponding distribution and for any two random variables $D_1, D_2$ in the space, $\Pr[D_1 \neq D_2] = |D_1 - D_2|$.
% \end{lemma}

\subsection{Communication Complexity} \label{section:communication}

Let $\mathcal{X}, \mathcal{Y}$ denote the set of possible inputs to the two players, who we name A and B.  In this
paper\footnote{The definitions we present here are equivalent to the classical definitions and are more convenient for
our proofs.}, we view a \emph{private coins protocol} for computing a function $f: \mathcal{X} \times \mathcal{Y}
\rightarrow \Z_K$ as a rooted tree with the following structure:
\begin{itemize}

\item Each non-leaf node is \emph{owned} by A or by B.
\item Each non-leaf node owned by a particular player has a set of children that are owned by the other player. Each of these children is labeled by a binary string, in such a way that this coding is prefix free: no child has a label that is a prefix of another child. 

\item Every node is associated with a function mapping $\mathcal{X}$ to distributions on children of the node and a function mapping $\mathcal{Y}$ to distributions on children of the node.

\item The leaves of the protocol are labeled by output values.

\end{itemize}

On input $x,y$, the protocol $\pi$ is executed as in \figureref{figure:pi}.

\begin{figure}[h!tb]
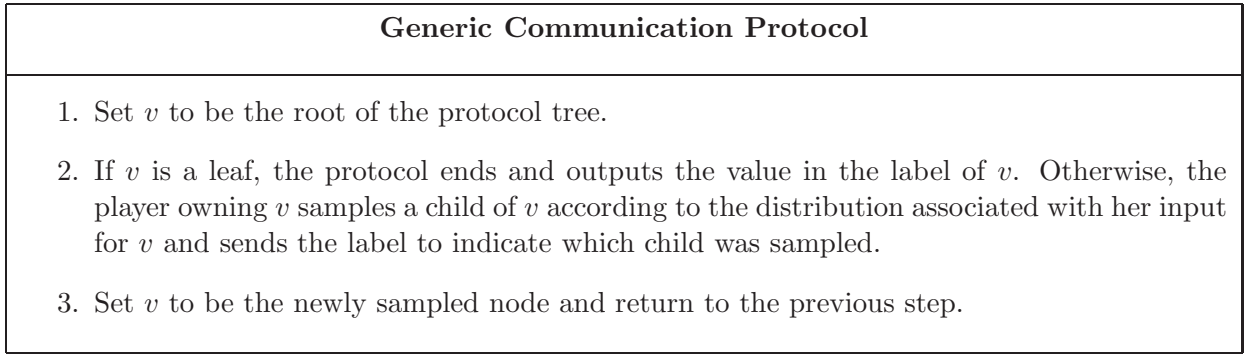

\begin{tabular}{|l|}
\hline
\begin{minipage}{\algwidth}
\vspace{1ex}
\begin{center}
\textbf{Generic Communication Protocol}
\end{center}
\vspace{0.5ex}
\end{minipage}\\
\hline
\begin{minipage}{\algwidth}
\vspace{1ex}
\begin{enumerate}
    \item Set $v$ to be the root of the protocol tree.
    \item If $v$ is a leaf, the protocol ends and outputs the value in the label of $v$. Otherwise, the player owning $v$ samples a child of $v$ according to the distribution associated with her input for $v$ and sends the label to indicate which child was sampled.
    \item Set $v$ to be the newly sampled node and return to the previous step.
\end{enumerate}
\vspace{0.3ex}
\end{minipage}\\
\hline
\end{tabular}
\caption{A communication protocol.}\label{figure:pi}
\end{figure}

A public coin protocol is a distribution on private coins protocols, run by first using shared randomness to sample an
index $r$ and then running the corresponding private coin protocol $\pi_r$. Every private coin protocol is thus a
public coin protocol. The protocol is called deterministic if all distributions labeling the nodes have support size
$1$.

\begin{definition}
The \emph{communication complexity} of a public coin protocol $\pi$, denoted $\CProt{\pi}$, is the maximum number of bits that can be transmitted in any run of the protocol.
\end{definition}

\begin{definition}
The \emph{number of rounds} of a public coin protocol is the maximum depth of the protocol tree $\pi_r$ over all choices of the public randomness. 
\end{definition}

Given a protocol $\pi$, $\pi(x,y)$ denotes the concatenation of the public randomness with all the messages that are
sent during the execution of $\pi$. We call this the \emph{transcript} of the protocol. We shall use the notation
$\pi(x,y)_j$ to refer to the $j$'th transmitted message in the protocol. We write $\pi(x,y)_{\leq j}$ to denote the
concatenation of the public randomness in the protocol with the first $j$ message bits that were transmitted in the
protocol. Given a transcript, or a prefix of the transcript, $v$, we write $\CProt{v}$ to denote the number of message
bits in $v$ (i.e.\ the length of the communication).

% 
% We often assume that every leaf in the protocol is at the same depth. We can do this since if some leaf is at depth
% less than the maximum, we can modify the protocol by adding dummy nodes which are always picked with probability $1$,
% until all leaves are at the same depth. This does not change the communication complexity.

\begin{definition}[Communication Complexity notation]
For a function $f: \mathcal{X} \times \mathcal{Y} \rightarrow \Z_K$, a distribution $\mu$ supported on $\mathcal{X}
\times \mathcal{Y}$, and a parameter $\rho > 0$,  $D^\mu_\rho(f)$ denotes the communication complexity of the cheapest
deterministic protocol for computing $f$ on inputs sampled according to $\mu$ with error $\rho$. $R_\rho(f)$ denotes
the cost of the best randomized public coin protocol for computing $f$ with error at most $\rho$ on \emph{every} input.
\end{definition}

For ease of notations, we shall sometimes use the shorthand $R(f)$ to denote $R_{1/3}(f)$.

We shall use the following theorem due to Yao:

\begin{theorem}[Yao's Min-Max] $R_\rho(f) = \max_\mu D^\mu_\rho(f)$.
\end{theorem}

Recall that the internal information cost $\IProt{\pi}{\mu}$ of a protocol $\pi$ is defined to be $I(\pi(X,Y);X|Y) + I(\pi(X,Y);Y|X)$.

\begin{lemma}\label{lemma:publicrandomness} Let $R$ be the public randomness used in the protocol $\pi$. Then $\IProt{\pi}{\mu} = \Ex{R}{\IProt{\pi_R}{\mu}}$
\end{lemma}
\begin{proof}
	By the chain rule (Proposition \ref{prop:infobreak}),
	\begin{align*}
	\IProt{\pi}{\mu} &= I(\pi(X,Y);X|Y) + I(\pi(X,Y);Y|X) \\
	&= I(R;X|Y) + I(R;Y|X) + I(\pi(X,Y);X|YR) + I(\pi(X,Y);Y|XR) \\
	&= I(\pi(X,Y);X|YR) + I(\pi(X,Y);Y|XR) \\
	&= \Ex{R}{\IProt{\pi_R}{\mu}}
\end{align*}
\end{proof}

A priori, one might believe that the internal information cost can be as large as twice the communication in a protocol. However, we can use the fact that each transmission only reveals information to one of the parties to bound it by the communication in the protocol:

\begin{lemma} \label{lemma:infocostbound} $\IProt{\pi}{\mu} \leq \CProt{\pi}$.
\end{lemma}
\begin{proof}
First, let us assume that $\pi$ is a private coin protocol. Let $\pi_i$ denote the $i$'th bit transmitted in the protocol. Then, by the chain rule,
	\begin{align*}
	\IProt{\pi}{\mu} &= I(\pi(X,Y);X|Y) + I(\pi(X,Y);Y|X) \\
	&= \sum_{i=1}^{\CProt{\pi}} I(\pi_i;X|\pi_1\pi_2\dotsc \pi_{i-1}Y) + I(\pi_i;Y|\pi_1\pi_2\dotsc \pi_{i-1} X) \\
\end{align*}
    Given any prefix $\gamma = \pi_1\dotsc \pi_{i-1}$, let $E_{\gamma}$ denote the event that the first $i-1$ bits of the transcript are equal to $\gamma$. Then we have

	\begin{align*}
	\IProt{\pi}{\mu} &= \sum_{i=1}^{\CProt{\pi}} \Ex{\gamma \getsr \pi_1 \dotsc \pi_{i-1}}{I(\pi_i;X| E_\gamma Y) + I(\pi_i;Y| E_\gamma X)}. \\
\end{align*}

Now we claim that $I(\pi_i;X| E_\gamma Y) + I(\pi_i;Y| E_\gamma X) \leq 1$. Each of these terms is individually bounded by $1$ since $\pi_i$ contains only one bit. If $\gamma$ is such that it is the first party's turn to transmit $\pi_i$, then for every fixing of $X$, $\pi_i$ is independent of $Y$, so $I(\pi_i;Y| E_\gamma X)=0$. On the other hand, if $\gamma$ is such that it is the second party's turn to transmit $\pi_i$, then for every fixing of $Y$, $\pi_i$ is independent of $X$, so $I(\pi_i;X| E_\gamma Y)=0$. Thus,

\begin{align*}
\IProt{\pi}{\mu} &= \sum_{i=1}^{\CProt{\pi}} \Ex{\tau \getsr \pi_1 \dotsc \pi_{i-1}}{I(\pi_i;X| E_\tau Y) + I(\pi_i;Y| E_\tau X)} \leq \CProt{\pi}. \\
\end{align*}

If $\pi$ involves public randomness, then by Lemma \ref{lemma:publicrandomness}, we have that $\IProt{\pi}{\mu} = \Ex{R}{\IProt{\pi_R}{\mu}} \leq \CProt{\pi}$, where $R$ denotes the public randomness of $\pi$.

\end{proof}

A version of the following theorem was proved in \cite{BaryossefJKS04}. Here we need a slightly stronger version (alluded to in a remark in \cite{BarakBCR10}):

\begin{theorem}\label{theorem:realtoinfo} For every $\mu,f,\rho$ there exists a protocol $\tau$ computing $f$ on inputs
drawn from $\mu$ with probability of error at most $\rho$ and communication at most $\mathsf{D}^{\mu^n}_\rho(f^n)$ such
that $\IProt{\tau}{\mu} \leq \frac{ D^{\mu^n}_\rho(f^n)}{ n}$.
	%\[ \mathsf{IC}^{\mu,\mathsf{D}^{\mu^n}_\rho(f^n)}_{\rho +\e}(f) \leq \frac{3 D^{\mu^n}_\rho(f^n)}{\e n},\]	
\end{theorem}

Since this theorem is subsumed by Theorem \ref{theorem:realtoinfo2} below, we do not give the details of its proof.

For our results on amortized communication complexity, we need the following definition: we shall consider the problem of computing $n$ copies of $f$, with error $\rho$ in each coordinate of the computation, i.e. the computation must produce the correct result in any single coordinate with probability at least $1-\rho$. We denote the communication complexity of this problem by $D^{\mu,n}_\rho(f)\le \mathsf{D}^{\mu^n}_\rho(f^n)$. Formally,

\begin{definition}
\label{def:amortized}
Let $\mu$ be a distribution on $X\times Y$ and let $0<\rho<1$. We denote by $D^{\mu,n}_\rho(f)$ the distributional complexity of computing $f$ on each of $n$ independent pairs of inputs drawn from $\mu$, with probability of failure at most $\mu$ on each of the inputs. 
\end{definition}

The result above can actually be strengthened:

\begin{theorem}\label{theorem:realtoinfo2} For every $\mu,f,\rho$, let $\pi$ be a protocol realizing $D^{\mu,n}_\rho(f)$. Then there exists a protocol $\tau$ computing $f$ on inputs
drawn from $\mu$ with probability of error at most $\rho$ such that $\CProt{\tau} = \CProt{\pi}$ and $\IProt{\tau}{\mu} \leq \frac{\IProt{\pi}{\mu^n}}{n} \leq \frac{ D^{\mu,n}_\rho(f)}{n}$.
	%\[ \mathsf{IC}^{\mu,\mathsf{D}^{\mu^n}_\rho(f^n)}_{\rho +\e}(f) \leq \frac{3 D^{\mu^n}_\rho(f^n)}{\e n},\]	
\end{theorem}
\begin{proof}

	First let us assume that $\pi$ only uses private randomness. The protocol $\tau(x,y)$ is defined as follows.
	\begin{enumerate}
		\item The parties publicly sample $J$, a uniformly random element of the set $\{1,2,3,\dotsc,n\}$.
		\item The parties publicly sample $X_1,\dotsc,X_{J-1}$ and $Y_{J+1},\dotsc,Y_n$.
		\item The first party privately samples $X_{J+1},\dotsc,X_n$ conditioned on the corresponding $Y$'s. Similarly, the second party privately samples $Y_1,\dotsc, Y_{J-1}$.
		\item The parties set $X_J = x$, $Y_J=y$, and run the protocol $\pi$ on inputs $X_1,\dotsc,X_n,Y_1,\dotsc,Y_n$. They output the result computed for the $J$'th coordinate.
	\end{enumerate}
	
	Observe that $\CProt{\tau} = \CProt{\pi}$, and the probability of making an error in $\tau$ is bounded by $\rho$. It only remains to bound $\IProt{\tau}{\mu} = I(X; \tau | Y) + I(Y; \tau | X)$. Let us bound the first term.
	\begin{align*}
	I (X; \tau |Y) &\leq I(X; \tau Y_1,\dotsc, Y_n |Y)\\
	&= I(X; J X_1\dotsc X_{J-1} Y_{1}\dotsc Y_n \pi | Y) \\
	&= I(X; J X_1\dotsc X_{J-1} Y_{1}\dotsc Y_n  | Y) + I(X_J; \pi | J X_1\dotsc X_{J-1} Y_{1}\dotsc Y_n )  \\
	&= I(X_J; \pi | J X_1\dotsc X_{J-1} Y_{1}\dotsc Y_n ) 
\end{align*}
	
	where the final equality is from the fact that $J, X_1,\dotsc,X_{J-1}, Y_{1},\dotsc,Y_n$ are all independent of $X,Y$, conditioned on every fixing of $Y$. % Since information can only be non-negative, we have
	% 	\begin{align*}
	% 		I (X; \tau |Y) &\leq I(X_J; \pi Y_1\dotsc Y_{J-1} | J X_1\dotsc X_{J-1} Y_{J}\dotsc Y_n )\\
	% 		& = I(X_J; Y_1\dotsc Y_{J-1} | J X_1\dotsc X_{J-1} Y_{J}\dotsc Y_n ) + I(X_J; \pi | J X_1\dotsc X_{J-1} Y_{1}\dotsc Y_n )\\
	% 		& = I(X_J; \pi | J X_1\dotsc X_{J-1} Y_{1}\dotsc Y_n ),
	% 	\end{align*}
	% 	where again we used the fact that the $J$'th coordinate is independent of all other coordinates.
	
	Expanding the expectation according to $J$, we get by the Chain Rule:
	\begin{align*}
		I (X; \tau |Y) &\leq (1/n) \sum_{j=1}^n I(X_j; \pi | X_1\dotsc X_{j-1} Y_{1}\dotsc Y_n ) \\
		&= I(X_1\dotsc X_n; \pi |Y_1 \dotsc Y_n)/n
	\end{align*}
	
	Similarly, we can bound $I(Y; \tau |X) \leq I(Y_1 \dotsc Y_n; \pi | X_1\dotsc X_{n}) / n$, and thus $\IProt{\tau}{\mu} \leq \IProt{\pi}{\mu^n}/n \leq \CProt{\pi}/n$, by Lemma \ref{lemma:infocostbound}.

	If $\pi$ uses public randomness $R$, then denote by $\tau_R$ the protocol induced for each fixing of $R$. Then $\IProt{\tau}{\mu} = \Ex{R}{\IProt{\tau_R}{\mu}} \leq \Ex{R}{\IProt{\pi_R}{\mu^n}/n } \leq \CProt{\pi}/n$.
	
\end{proof}

% 
% \begin{remark}[Information content of private vs. public coins protocols.]
% Another way to view the difference between public coins and private coins protocols is that the public randomness is
% considered part of the protocol's transcript. But even if the randomness is short compared to the overall communication
% complexity, making it public can have a dramatic effect on the information content of the protocol. (As an example,
% consider a protocol where one party sends a message of $x \oplus r$ where $x$ is its input and $r$ is random. If the
% randomness $r$ is private then this message has zero information content. If the randomness is public then the message
% completely reveals the input. This protocol may seem trivial since its communication complexity is larger than the
% input length, but in fact we will be dealing with exactly such protocols, as our goal will be to ``compress''
% communication of protocols that have very large communication complexity, but very small information content.)
% \end{remark}

\section{Proof of Theorem \ref{theorem:main}}
% 
% \section{Summary of definitions}
% 
% We have two players $A$ and $B$. The first player $A$ has a distribution $P$ over a universe $\cU$ that is not too
% different from a  distribution $Q$ known to the second player $B$. % The goal is 
% % for both players to sample from $P$ with as little communication as possible. 
% % 
% The fact that $Q$ is close to $P$ is formalized using information divergence $\Div{P}{Q}$.
% 
% \begin{definition}
% \label{def:D}
% $D:=\displaystyle{\Div{P}{Q}=\sum_{x\in \cU} P(x) \log\frac{P(x)}{Q(x)}}$.
% \end{definition}
% 
% We will see that the divergence $\Div{P}{Q}$ is the ``right" notion of closeness for the purpose of
% sampling. Note that the function $D$ is not symmetric. In fact, it is possible for $D(Q\|P)$ to be infinite
% while $\Div{P}{Q}$ is finite. We assume that both players have access to a shared public randomness source. 

We shall prove a stronger version of Theorem \ref{theorem:main}.

\begin{theorem}
\label{thm:Main}
Suppose that player $A$ is given a distribution $P$ and player $B$ is given a distribution $Q$ over a universe $\cU$. 
There is a protocol such that 
at the end of the protocol:
\begin{itemize}
\item 
player $A$ outputs an element $a$ distributed according to $P$;
\item 
player $B$ outputs an element $b$ such that for each $x$, $\P[b=a|~a=x]>1-\e$.
\item the communication in the protocol is bounded by $ \log P(a)/Q(a) + \log 1/\e + \log\log 1/\e +  5\sqrt{\log P(a)/Q(a)} + 9$.
\end{itemize}
\end{theorem}

\noindent
Note that the second condition implies in particular that  player $B$ outputs an element $b$ such that $b=a$ with probability $>1-\e$.
The protocol requires no prior knowledge or assumptions on $\Div{P}{Q}$. 

\begin{proof}
The protocol runs as follows. Both 
parties interpret the shared random tape as a sequence of uniformly selected elements $\{a_i\}_{i=1}^\infty = \{(x_i,p_i)\}_{i=1}^\infty$ from 
the set $\cA:= \cU \times [0,1]$. Denote the subset
$$
\cP := \{(x,p)~:~P(x)>p\}
$$
of $\cA$ as the set of points under the histogram of the distribution $P$. Similarly, define 
$$\cQ := \{(x,p)~:~Q(x)>p\}.$$
For a constant $C\ge 1$ we will define the $C$-multiple of $\cQ$ as
$$
C\cdot \cQ := \{(x,p)\in \cA~:~(x,p/C)\in \cQ\}.
$$

We will also use a different part of the shared random tape to obtain a sequence of 
random hash functions $h_i : \cU \rightarrow \{0,1\}$ so that for any $x\neq y\in \cU$, $\P[h_i(x)=h_i(y)]=1/2$.

\begin{figure}[ht] 
\begin{center} \includegraphics[angle=0,scale=.8]{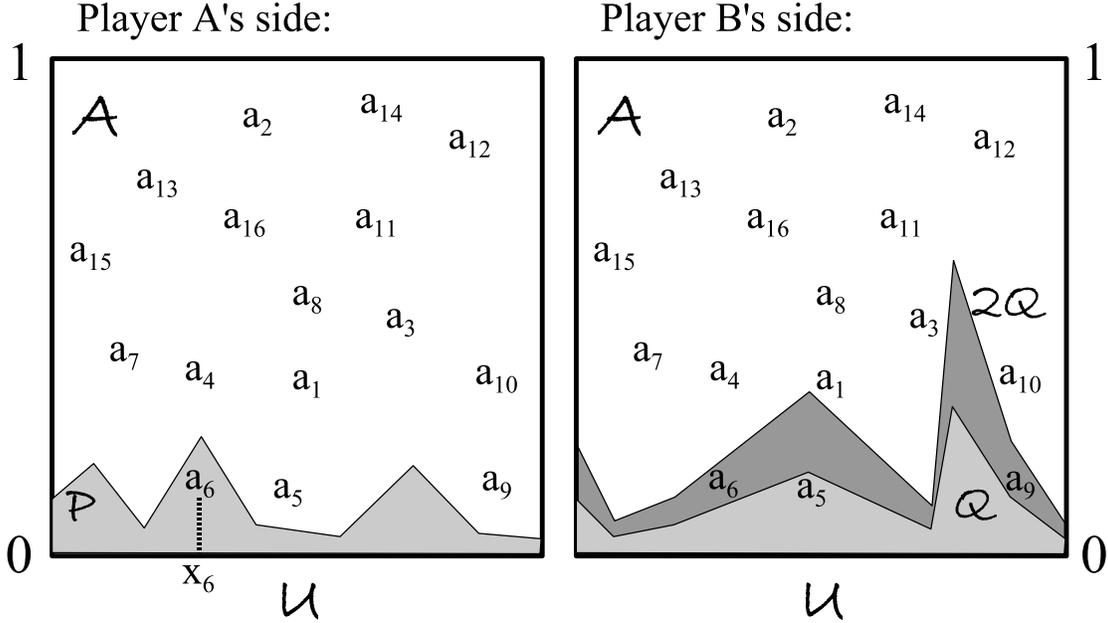}
\caption{An illustration on the execution of the protocol. The elements $a_i$ are selected uniformly
from $\cA = \cU\times [0,1]$. 
 The first $a_i$ to fall in $\cP$ is $a_6$, and 
thus player $A$ outputs $x_6$. Player $A$ sends hashes of $a_6$, which do not match the hashes of 
$a_5$, the only $a_i$ in $\cQ$. Player $B$ responds `failure', and considers surviving elements in $2\cQ$, which 
are $a_6$ and $a_9$. After a few more hashes from $A$, $a_6$ is selected by $B$ with high probability.}
\label{fig:sampling}
\end{center}
\end{figure}

We are now ready to present the protocol:

\begin{enumerate}
\item 
Player $A$ selects the first index $i$ such that $a_i=(x_i,p_i)\in \cP$, and outputs $x_i$;
\item 
Player $A$ uses $1+\lceil \log\log 1/\e \rceil$ bits to send Player $B$ the binary encoding of $k := \lceil i/|\cU|\rceil$ (if $k$ is too large, Player $A$ sends an arbitrary string);
\item 
For all $t$, set parameters $C_t:= 2^{t^2}$, $s_t = 1 + \lceil \log 1/\e \rceil + (t+1)^2$;
\item 
Repeat, until Player $B$ produces an output, beginning with iteration $t=0$: \label{st:4}
\begin{enumerate}
\item
Player $A$ sends the values of all hash functions $h_j(x_i)$ for $1 \leq j \leq s_t$, that have not previously been sent.
\item 
if there is an $a_r=(y_r,q_r)$  with $r\in\{(k-1)\cdot |\cU|+1,\ldots,k\cdot |\cU|\}$ in $C_t\cdot \cQ$ such that $h_j(y_r)=h_j(x_i)$ for $1 \leq j \leq s_t$, Player $B$ responds `success' and outputs $y_r$; if 
there is more than one such $a_r$, player $B$ selects the first one;
\item 
otherwise, Player $B$ responds `failure', and the parties increment $t$ and repeat.
\end{enumerate}
\end{enumerate}
The output of Player $A$ is distributed according to the distribution $P$, and further, the output is independent of $k$. To see this, note that the output is independent of whether or not $k> s$, for every $s$.

Fix a choice of $i$ and the pair $(x_i,p_i)$ by Player $A$. Step \ref{st:4} of the protocol is guaranteed to terminate when $t^2 \ge \log P(x_i)/Q(x_i)$ since $a_i$ belongs
to $\frac{P(x_i)}{Q(x_i)} \cdot \cQ$. Denote $T:=\left\lceil \sqrt{\log P(x_i)/Q(x_i)} \right\rceil$. By iteration $T$, Player $A$ will have sent $s_T$ bits in Step \ref{st:4}, and Player $B$ will have sent $T+1$ bits. Thus the amount of communication in Step \ref{st:4} is bounded by 
\begin{align*}
s_T + T+1 &= 1 + \lceil \log 1/\e \rceil + (T+1)^2 + T+1 \\
& \leq (\sqrt{\log P(x_i)/Q(x_i)} + 2)^2 + \sqrt{\log P(x_i)/Q(x_i)} + \lceil \log 1/\e \rceil + 3\\
& = \log P(x_i)/Q(x_i) + 5\sqrt{\log P(x_i)/Q(x_i)} +  \lceil \log 1/\e \rceil + 7,
\end{align*}
which shows that the total communication is at most $$\log P(x_i)/Q(x_i) + 5 \sqrt{\log P(x_i)/Q(x_i)} + \log 1/\e + \log\log 1/\e + 9.$$  

It only remains to show that Player $B$ outputs the same $x_i$ with probability $>1-\e$. We start with the following claim. 

\begin{claim}
\label{cl:1} 
For each $n$, $\P[k> n] < e^{-n}$.
\end{claim}

\begin{proof}
For each $n$, we have 
$$
\P[k>n] = \P[a_i \notin \cP \mbox{ for } i=1,\ldots,n \cdot |\cU|] = (1-1/|\cU|)^{|\cU|\cdot  n} <e^{-n}.
$$
\end{proof}

Thus the probability that the binary encoding of $k$ exceeds $1+ \lceil \log\log 1/\e \rceil$ bits is less than $e^{ -2 \cdot 2^{\lceil \log\log 1/\e \rceil}} \leq  \e/2$. It remains 
to analyze Step \ref{st:4} of the protocol. We say that an element $a=(x,p)$ {\em survives} iteration $t$ if $a\in 2^{t^2}\cdot \cQ$ and it 
satisfies $h_j(x) = h_j (x_i)$ for all $j=1,\ldots,s_t$ for this $t$.

 Note that the ``correct" element $a_i$ survives
iteration $t$ if and only if $2^{t^2} \ge P(x_i)/Q(x_i)$. 

\begin{claim}
\label{cl:2}
Let $E_{a_i}$ be the event that the element selected by player $A$ is $a_i$, which is the $i$-th element on the tape. 
Denote $k:=\lceil i/|\cU| \rceil$.
Conditioned on $E_{a_i}$, the probability that a different element $a_j$ with $j \in \{ (k-1)\cdot |\cU|+1,\ldots,k\cdot |\cU| \}$ survives iteration $t$ is bounded by $\e/2^{t+1}$.
\end{claim}

\begin{proof}
Without loss of generality we can assume that  $|\cU|\ge 2$, since for a singleton universe our sampling protocol will succeed trivially. 
This implies that for a uniformly selected $a\in \cA$, $\P[a  \notin \cP] \geq 1/2$, so: $$\P[a\in C_t\cdot \cQ|~a\notin \cP]\le \P[a\in C_t\cdot \cQ]/\P[a\notin \cP] \le 2\cdot \P[a\in C_t\cdot \cQ]\le 2C_t/|\cU|.$$ 
Denote $K:=k\cdot |\cU|$.
Conditioning on $E_{a_i}$, the elements $a_{K-|\cU|+1},\ldots,a_{i-1}$ are distributed uniformly on $\cA\setminus \cP$, and $a_{i+1},\ldots,a_K$ are 
distributed uniformly on $\cA$. For any such $j=K-|\cU|+1,\ldots,i-1$, and for any $C>0$,
$$
\P[a_j\in C\cdot \cQ] \le 2C/|\cU|.
$$
For such a $j$, surviving round $t$ means $a_j$ belonging to $2^{t^2}\cdot\cQ$ and agreeing with $a_i$ on $s_t = 1 + \lceil \log 1/\e \rceil+ (t+1)^2$ random hashes $h_1,\ldots,h_{s_t}$. The probability of this event is thus bounded by
\begin{align*}
\P[\text{$a_j$ survives round $t$}] &\le \P[a_j\in 2^{t^2}\cdot \cQ]\cdot 2^{-s_t} \\
&\le 2 \cdot 2^{t^2} \cdot 2^{-s_t} /|\cU|\\
&\leq 2^{1+t^2 - s_t } /|\cU| \\
&\leq 2^{-2t-1}\e/|\cU|.
\end{align*}
By taking a union bound over all $j=K-|\cU|+1,\ldots,K$, $j\neq i$, we obtain the bound of $\e/2^{2t+1} \leq \e/2^{t+1}$.
\end{proof}

\noindent
Thus for any $E_{a_i}$, the probability of Player $B$ to output anything other than $x_i$ conditioned on $E_{a_i}$ is $<\sum_{t=0}^\infty \e/2^{t+1} = \e$. 
 
\end{proof}

To get a bound on the expected amount of communication in the protocol, as in Theorem \ref{theorem:main}, note that
\begin{multline*}
{\bf E}_{x_i\sim P}  \left[ {\log P(x_i)/Q(x_i)} + 2+5 \sqrt{\log P(x_i)/Q(x_i)} \right] = 
\Div{P}{Q}+9+ 5 \cdot {\bf E}_{x_i\sim P} \sqrt{\log P(x_i)/Q(x_i)} \\ \le \Div{P}{Q}+2 + 5 \cdot \sqrt{{\bf E}_{x_i\sim P} {\log P(x_i)/Q(x_i)}} = 
\Div{P}{Q}+O(\Div{P}{Q}^{1/2}+1), 
\end{multline*}
where the inequality is by the concavity of $\sqrt{~~}$. This completes the proof.

\begin{remark} Note that if the parties are trying to sample many independent samples from distributions $P_1,P_2,\dotsc$, with the receiving party knowing $Q_1,\dotsc$, as in the setting of the Slepian-Wolf theorem, the analysis of the above protocol can easily be strengthened to show that the communication is the right amount with high probability. This is because the central limit theorem can be used to show that samples $x_i$ with higher than expected $P_i(x_i)/Q_i(x_i)$ are rare and do not contribute much to the communication on average (see for example Section~\ref{subsec:infoequals}).
\end{remark}

\begin{remark}
The sampling in the proof of Theorem \ref{thm:Main} may take significantly more than one round. In fact, the expected 
number of rounds is $\Theta(\sqrt{ \Div{P}{Q}})$. We suspect that the dependence of the number 
of rounds in the simulation on the divergence cannot be eliminated, since $\Div{P}{Q}$ is not known to the players ahead of time, and 
the only way to ``discover" it (and thus to estimate the amount of communication necessary to perform the sampling task)
is through interactive communication. By increasing the expected communication by a constant multiplicative factor, it 
is possible to decrease the expected number of rounds to $O(\log \Div{P}{Q})$.
%We defer the details to the full version of the paper.
\end{remark}
% 
% \noindent
% For technical reasons we will need the following easy extension of Theorem \ref{thm:Main}:
% 
% \begin{lemma}
% \label{lem:Event}
% In the setup of Theorem \ref{thm:Main} there is an event $E$ such that $\P[E]>1-\e$, and conditioned on $E$:
% \begin{itemize}
% \item 
% both parties output the same value distributed exactly according to $P$;
% \item 
% if $x$ is sampled, then the probability that the communication exceeds $\log P(x)/Q(x) +\log 1/\e +5\sqrt{\log P(x)/Q(x)}+2 + \ell$ is bounded by $e^{-\ell}$.
% \end{itemize}
% \end{lemma}
% 
% \begin{proof}
% Let $E'$ be the event when both parties output the same value (i.e. when the protocol succeeds). Since 
% the probability of success is $>1-\e$ conditioned on the value being output by Player $A$, there is an event $E\subset E'$
% such that $\P[\text{Player $B$ outputs $x$}|E]=P(x)$ for all $x$ and $\P[E]>1-\e$. 
% 
% It remains to see that the communication guarantee holds. The communication can always be bounded in terms of the element $x$ sampled. The probability that the communication exceeds $\ell$ in Step 1 is bounded by $e^{-\ell}$, and the communication for the remaining rounds is always at most $\log P(x)/Q(x) + \log (1/\e) + 5(\log P(x)/Q(x))^{1/2}+2$. This gives our final bounds.
% \end{proof}
% 

\section{Correlated Pointer Jumping}
Here we define the correlated pointer jumping problem, that is at the heart of several of our results. The input in this problem is a rooted tree such that

\begin{itemize}
\item Each non-leaf node is \emph{owned} by Player A or by Player B.
\item Each non-leaf node owned by a particular player has a set of children that are owned by the other player. Each of these children is labeled by a binary string, in such a way that this coding is prefix free: no child has a label that is a prefix of another child. 
\item Each node $v$ is associated with two distributions on its children: $\childa{v}$ known to Player A and $\childb{v}$ known to Player B.
\item The leaves of the tree are labeled by output values.
\end{itemize}

The number of rounds in the instance is the depth of the tree. 

The goal of the problem is for the players to sample the leaf according to the distribution that is obtained by sampling each child according to the distribution specified by the owner of the parent. We call the distribution of this path, the \emph{correct} distribution. We give a way to measure the correlation between the knowledge of the two parties in the problem.

For every non-root vertex $w$ in the tree whose parent is $v$, define the \emph{divergence cost} of $w$ as
 	\begin{align*}
		\DivProt{w} = \begin{cases}
		\log \left( \frac{\childa{v}(w)}{\childb{v}(w)} \right ) & \text{if v is owned by Player A}\\\\
		\log \left( \frac{\childb{v}(w)}{ \childa{v}(w)} \right) & \text{if v is owned by Player B}
	\end{cases}
	\end{align*}
The divergence cost of the root is set to $0$.

Given a path $T$ that goes from the root to a leaf in the tree, the divergence cost of the path, denoted $\DivProt{T}$ is the sum of the divergence costs of the nodes encountered on this path. Finally, the divergence cost of the instance $F$, denoted $\DivProt{F}$ is the expected sum of divergence costs of the vertices encountered in the correct distribution on paths.% 
% 
% \begin{definition}[Divergence Cost] The divergence cost of a correlated pointer jumping instance whose root is $v$, denoted $\DivProt{F}$, is recursively defined as follows:
% 	\begin{align*}
% 		\DivProt{F} = \begin{cases}
% 		 0 & \text{if the tree has depth 0} \\
% 		\Div{\child(v)_x}{\child(v)_y} + \Ex{w \getsr \child(v)_x}{\DivProt{F_w}} & \text{if v is owned by $P_x$}\\
% 		\Div{\child(v)_y}{\child(v)_x} + \Ex{w \getsr \child(v)_y}{\DivProt{F_w}} & \text{if v is owned by $P_y$}
% 	\end{cases}
% 	\end{align*}
% \end{definition}

We can use our sampling lemma to solve the correlated pointer jumping problem, with communication bounded by the divergence cost:

\begin{theorem}
 \label{theorem:pointer}
There is a protocol that when given a $k$-round correlated pointer jumping instance $F$, can sample a path $T$ such that there is an event $E$, with $\P[E] >1-k\ve$, and conditioned on $E$, 
\begin{itemize}
	\item the parties both output the same sampled path $T$ that has the correct distribution 
	\item the communication in the protocol is bounded by $\DivProt{T} + 2 k \log(1/\ve)+  5 \sqrt{k\DivProt{T}}+ 9k$.
\end{itemize}
\end{theorem}
\begin{proof}
The protocol for sampling the path is obtained simply by repeatedly running the protocol from Theorem \ref{thm:Main}. In each step, the parties sample the correct child. For each round $i$ let $E_i$ denote the event that the parties are consistent after round $i$. When $E_i$ occurs, the sampled vertex has the correct distribution, and $\Pr[E_i] > 1-\ve$. Define $E$ to be the intersection of the events $E_i$. Then $\Pr[E] > 1-k \ve$. Conditioned on $E$, the sampled path has the correct distribution. Moreover, the if the sampled path is $T = v_0, v_1,\dotsc, v_k$, then by Theorem \ref{thm:Main}, the communication in the protocol is at most

\begin{align*}
	&\sum_{i=1}^k \left (\DivProt{v_i} + \log(1/\ve) + \log \log(1/\ve) + 5 \sqrt{\DivProt{v_i}} + 9 \right ) \\
	&\leq \sum_{i=1}^k \left (\DivProt{v_i} \right) + 2 k \log( 1/\ve) + 5 \sqrt{ k \cdot \sum_{i=1}^k \DivProt{v_i}} + 9k\\
	& = \DivProt{T} + 2 k \log( 1/\ve) + 5 \sqrt{ k \cdot \DivProt{T} } + 9k,
\end{align*}
where the inequality is by the Cauchy-Schwartz inequality.

\end{proof}

A key fact is that both the internal and external information cost of a protocol can be used to bound the expected divergence cost of an associated distribution on correlated pointer jumping instances. Since, in this work, we only require the connection to internal information cost, we shall restrict our attention to it.

Given a public coin protocol with inputs $X,Y$ and public randomness $R$, for every fixing of $x,y,r$, we obtain an instance of correlated pointer jumping. The tree is the same as the protocol tree with public randomness $r$. If a node $v$ at depth $d$ is owned by Player A, let $M$ be the random variable denoting the child of $v$ that is picked. Then define $\childa{v}^x$ so that it has the same distribution as $M|~X=x,\pi(X,Y)_{\leq d}=rv$, and $\childb{v}^y$ so it has the same distribution as $M|~Y=y,\pi(X,Y)_{\leq d}=rv$. We denote this instance of correlated sampling by $F_{\pi}(x,y,r)$. Let $\mu$ denote the distribution on $X,Y$. Next we relate the average divergence cost of this instance to the internal information cost of $\pi$:

\begin{lemma}\label{lemma:infovsdiv} $\Ex{X,Y,R}{\DivProt{F_\pi(x,y,r)}} = \IProt{\pi}{\mu}$
\end{lemma}
\begin{proof}
	We shall prove that for every $r$, $\Ex{X,Y}{\DivProt{F_\pi(x,y,r)}} = \IProt{\pi_r}{\mu}$. The proof can then be completed by \lemmaref{lemma:publicrandomness}.
	
	So without loss of generality, assume that $\pi$ is a private coin protocol, and write $F(x,y)$ to denote the corresponding divergence cost.
	We proceed by induction on the depth of the protocol tree of $\pi$. If the depth is $0$, then both quantities are $0$. For the inductive step, 
	without loss of generality, assume that Player A owns the root node $v$ of the protocol. Let $M$ denote the child of the root that is sampled during the protocol, and let $F(x,y)_m$ denote the divergence cost of the subtree rooted at $m$. Then 
	
	\begin{align}
	  \Ex{X,Y}{\DivProt{F(x,y)}} = \Ex{x,y,m \getsr X,Y,M}{ \log (\childa{v}^x (m)/\childb{v}^y(m)) + \DivProt{F(x,y)_m}}  \label{eqn:divcost}
	\end{align}
	
	Since for every $x,y$, $M|xy$ has the same distribution as $M|x$, \propositionref{prop:divinfo} gives that the first term in \equationref{eqn:divcost} is exactly equal to $I(X;M|Y) = I(X;M|Y) + I(Y;M|X)$. The second term is $\Ex{M}{\Ex{X,Y | M}{\DivProt{F(X,Y)_M}}}$. For each fixing of $M=m$, the inductive hypothesis shows that the inner expectation is equal to $I(X;\pi(X,Y)|Ym) + I(Y;\pi(X,Y)|Xm)$. Together, these two bounds imply that 
	\begin{align*} 
		&\Ex{X,Y}{\DivProt{F(x,y)}} \\
		&=  I(X;M|Y) + I(Y;M|X) + I(X;\pi(X,Y)|YM) + I(Y;\pi(X,Y)|XM) \\
		&= \IProt{\pi}{\mu}
	\end{align*}
\end{proof}

\section{Applications}
In this section, we use \theoremref{theorem:pointer} to prove a few results about compression and direct sums. 

\subsection{Compression and Direct sum for bounded-round protocols} \label{subsec:compress}
Here we prove our result about compressing bounded round protocols (Corollary \ref{corollary:compress}). We shall need the following lemma.

\begin{lemma}
\end{lemma}

\begin{proof}[Proof of \corollaryref{corollary:compress}]
	The proof follows by applying our sampling procedure to the correlated pointer jumping instance $F_\pi(x,y,r)$. For each fixing of $x,y,r$, define the event $G_{x,y,r}$ to be the event $E$ from \theoremref{theorem:pointer}. Then we have that $\P[G]> 1-k\ve$. Conditioned on $G$, we sample from exactly the right distribution, and the expected communication of the protocol is 
	\begin{align*}
	&\Ex{X,Y,R}{\DivProt{F_\pi(X,Y,R)} +2 k\log(1/\ve)+  O(\sqrt{k\DivProt{F_\pi(X,Y,R)}}+k)} \\
	&\leq \Ex{X,Y,R}{\DivProt{F_\pi(X,Y,R)}} + 2k \log(1/\ve) + O \left (\sqrt{\Ex{X,Y,R}{ k\DivProt{F_\pi(X,Y,R)}} } + k \right),
\end{align*}
   where the inequality follows from the concavity of the square root function.
   By \lemmaref{lemma:infovsdiv}, this proves that the expected communication conditioned on $G$ is $\IProt{\pi}{\mu} + 2k \log(1/\ve) + O \left (\sqrt{k\IProt{\pi}{\mu} } + k \right)$. 
     
\end{proof}

\subsection{Information = amortized communication} \label{subsec:infoequals}

In this section we will show that Theorem \ref{theorem:pointer} reveals a tight connection between 
the amount of information that has to be revealed by a protocol computing a function $f$ and 
the amortized communication complexity of computing many copies of $f$. Recall that 
$\IProt{f,\rho}{\mu}$ denotes the smallest possible internal information cost of any 
protocol computing $f$ with probability of failure at most $\rho$ when the inputs are drawn from the distribution $\mu$.
Observe that $\IProt{f,\rho}{\mu}$ is an infimum over all possible protocols and may not be achievable by any individual 
protocol. It is also clear that $\IProt{f,\rho}{\mu}$ may only increase as $\rho$ decreases. 

We first make the following simple observation. 

\begin{claim}
\label{cl:ICcont}
For each $f$, $\rho$ and $\mu$, 
$$
\lim_{\al\rightarrow \rho} \IProt{f,\al}{\mu} = \IProt{f,\rho}{\mu}
$$
\end{claim}

\begin{proof}
	The idea is that if we have any protocol with internal information cost $I$, error $\delta$ and input length $\ell$, for every $\ve$ we can decrease the error to $(1-\ve)\delta$ at the cost of increasing the information by at most $\ve \cdot \ell$ just by using public randomness to run the original protocol with probability $1-\ve$, and with probability $\ve$, run the trivial protocol where the players simply exchange their inputs. Thus as $\al$ tends to $\rho$, the information cost of the best protocols must tend to each other. 
\end{proof}

Next we define the amortized communication complexity of $f$. We define it to be the cost of computing $n$ copies of $f$ with error $\rho$ in each coordinate, divided by $n$. Note that 
computing $n$ copies of $f$ with error $\rho$ in each coordinate is in general an easier task than computing $n$ copies of $f$ with probability of success  $1-\rho$ for all copies. We use the notation $D^{\mu,n}_\rho(f)$ to denote the communication complexity for this task, when the inputs for each coordinate are sampled according to $\mu$. $D^{\mu,n}_\rho(f)$ was formally defined in Definition \ref{def:amortized}.

It is trivial to see in this case that $D^{\mu,n}_\rho(f)\le n\cdot D^{\mu}_\rho(f)$. The amortized communication complexity 
of $f$ with respect to $\mu$ is the limit $$\AProt{f^\mu_\rho}:=\lim_{n\rightarrow \infty}D^{\mu,n}_\rho(f)/n,$$ when the 
limit exists. We prove an exact equality between amortized communication complexity and the information cost:

\begin{theorem}
\label{thm:amortized}
$$\AProt{f^\mu_\rho} = \IProt{f,\rho}{\mu}.$$
\end{theorem} 

\begin{proof} There are two directions in the proof:

\smallskip
\noindent
{\underline{ $\AProt{f^\mu_\rho} \ge  \IProt{f,\rho}{\mu}$.}} This is a direct consequence of \theoremref{theorem:realtoinfo2}.

\smallskip
\noindent
{\underline{ $\AProt{f^\mu_\rho} \le  \IProt{f,\rho}{\mu}$.}} Let $\delta>0$. We will show that 
$D^{\mu,n}_\rho(f) /n <  \IProt{f,\rho}{\mu}+\delta$ for all sufficiently large $n$.

By Claim \ref{cl:ICcont} there is an $\al<\rho$ such that $\IProt{f,\al}{\mu}<\IProt{f,\rho}{\mu}+\delta/4$.
Thus there is a protocol $\pi$ that computes $f$ with error $<\al$ with respect to $\mu$ and 
that has an internal information cost bounded by $I:=\IProt{f,\rho}{\mu}+\delta/4$. 

For every $n$, denote by $\pi^n$ the protocol that takes $n$ pairs of inputs $X^n, Y^n$ and executes in parallel, sending the first bits of each copy in the first round, and then the second bits in the second round and so on. Thus $\pi^n$ has $\CProt{\pi}$ rounds, and communication complexity $n \CProt{\pi}$. Further, $\pi^n$ computes $n$ copies of $f$ as per Definition~\ref{def:amortized} with error bounded by $\al$. 

We shall obtain our results by compressing $\pi^n$.

Let $^i\pi$ denote the transcript of the $i$'th copy, and let $X_i, Y_i$ denote the $i$'th inputs. Then observe that for all $i$, $(X_i,Y_i,^i\pi)$ are mutually independent of each other. Indeed, this implies that $\IProt{\pi^n}{\mu^n} = \sum_{i=1}^n \IProt{^i\pi}{\mu} = n\IProt{\pi}{\mu}$. On the other hand, compressing $\pi^n$ incurs a per round overhead that is still dependent only on $\CProt{\pi}$. 

Let $T^n$ denote the random variable of the path sampled in $\pi^n$. Let $T_1,\dotsc, T_n$ denote the random variables of the $n$ paths sampled in the individual copies of $\pi$. Then, since each protocol runs independently, $\Expect{\DivProt{T^n}} = \sum_{i=1}^n \Expect{\DivProt{T}}$. Indeed, each vertex in the protocol tree of $\pi^n$ corresponds to an $n$-tuple of vertices of $\pi$, and if $w$ corresponds to the vertices $(^1w, \dotsc, ^nw)$, with parents $v = (^1v,\dotsc,^nv)$ owned by Player A, then \[\DivProt{w} = \log \left ( \frac{\childa{v}(w)}{\childb{v}(w)} \right ) = \log \left(   \frac{\prod_{i=1}^n \childa{v_{i}}(w_{i})}{\prod_{i=1}^n \childb{v_{i}}(w_{i})}\right) = \sum_{i=1}^n \log \left(  \frac{\childa{v_i}(w_i)}{\childb{v_i}(w_i)}\right) = \sum_{i=1}^n \DivProt{w_i}.\]

By Lemma \ref{lemma:infovsdiv}, $\Expect{\DivProt{T}} = \IProt{\pi}{\mu}$. Thus, by the central limit theorem, for $n$ large enough, \[\Pr[\DivProt{T^n} \geq  n \cdot (\IProt{\pi}{\mu} + \delta/4)] < (\rho - \al)/2.\] 

We use Theorem \ref{theorem:pointer} to simulate $\pi^n$, with error parameter $\e = (\rho - \al)/2$ and truncate the protocol after \[ n \cdot (\IProt{\pi}{\mu} + \delta/4) + 5 \sqrt{\CProt{\pi}\cdot n \cdot (\IProt{\pi}{\mu} + \delta/4)} + 2 \log(1/\e) + 9 \cdot \CProt{\pi} \] bits of communication. The new protocol thus has error $< \al + \rho - \al = \rho $. On the other hand, for $n$ large enough, the per copy communication of this protocol is at most $\IProt{\pi}{\mu} + \delta/2$ as required.

% 
% 
% The internal information cost of $\pi^n$ is $n\cdot I$, and by Theorem \ref{theorem:pointer} we can
% simulate $\pi^n$ with a total error $\ve<\rho-\al$ using 
% $$
% C_n := n\cdot I + C\cdot \log 1/\ve + O(\sqrt{C\cdot I\cdot n}+C)
% $$
% bits of communication. The total additional error is $\ve$ and hence the new protocol makes at most an error $\al+\ve<\rho$ on
% each copy of $f$. Hence $D^{\mu,n}_\rho(f)\le C_n$.
% By letting $n$ be large enough (with respect to $C$ and $1/\ve$) we see that we can make $D^{\mu,n}_\rho(f)\le C_n < n\cdot I + n\delta/2$,
% thus completing the proof. 
\end{proof}

\subsection{A complete problem for direct sum}
\label{sec:53}

Let $f^n$ denote the function mapping $n$ inputs to $n$ outputs according to $f$. We will show that the promise version of the correlated pointer jumping problem is complete for 
direct sum.  In other words, if near-optimal protocols for correlated pointer jumping exist, then 
direct sum holds for all promise problems. On the other hand, if there are no near-optimal protocols for correlated
pointer jumping, then direct sum  fails to hold, with the problem itself as the counterexample. 
Thus any proof of  direct sum for randomized communication complexity must give (or at least demonstrate existence) of
near-optimal protocols for the problem. 

We define the $\CPJ(C,I)$ promise problem as follows.

\begin{definition}
\label{defn:CPJ}
The $\CPJ(C,I)$ is a promise problem, where the players are provided with a {\em binary} instance\footnote{Each vertex has degree $2$.} $F$ of a $C$-round 
pointer jumping problem, i.e. player A is provided with the 
distributions $\child(v)_x$ and player B is provided with the 
distributions $\child(v)_y$ for each $v$, with the following additional guarantees:
\begin{itemize}
\item 
the divergence cost $\DivProt{F} \le I$; 
\item 
let $\mu_F$ be the correct distribution on the leafs of $F$; each leaf $z$ of $F$ are labeled with $\ell(z)\in\{0,1\}$ so
that there is a value $g = g(F)$ such that $\P_{z\getsr\mu_F}[\ell(z)=g(F)]>1-\ve$, for some small $\ve$.
The goal of the players is to output $g(F)$ with probability $>1-2\ve$.  
\end{itemize}
\end{definition}

Note that players who know how to sample from $F$ can easily solve the $\CPJ$ problem. It follows from \cite{BarakBCR10}
that:

\begin{theorem}
If $\CPJ(C,I)$ has a randomized protocol that uses $T(C,I):=\CFunc{\CPJ(C,I)}$ communication, so that $T(C,C/n)<C/k(n)$, then for each 
$f$, $$ \CFunc{f^n} = \Omega(k(n)\cdot \CFunc{f}).$$
\end{theorem}

In \cite{BarakBCR10} a bound of $T(C,I)=\tilde{O}(\sqrt{C\cdot I})$ is shown, which implies $ \CFunc{f^n} = \tilde{\Omega}(\sqrt{n}\cdot \CFunc{f})$
for any $f$. Using Theorem~\ref{theorem:pointer} we are able to prove the converse direction.  

\begin{theorem}
\label{thm:CPJ}
 For any $C>I>0$, set $n:=\lfloor C/I\rfloor$, then
$$
\CFunc{\CPJ(C,I)^n} = O(C \log (n C/\ve)). 
$$
\end{theorem}

Thus, if there are parameters $C$ and $n$ such that $\CPJ(C,C/n)$ cannot be solved using $I=C/n$ communication, i.e. 
$T(C,C/n)>C/k(n)\gg C/n$, then $\CPJ(C,C/n)$ is a counterexample to direct sum, i.e. 
$$
\CFunc{\CPJ(C,I)^n} = O(C \log nC/\ve) = \tilde{O}(C) = \tilde{O}(k(n)\CFunc{\CPJ(C,C/n)}) = o(n \cdot \CFunc{\CPJ(C,C/n)}).
$$

\begin{proof}(of Theorem \ref{thm:CPJ}) 
We solve $\CPJ(C,I)^n$ by taking $m:= n\log n$ copies of the $\CPJ(C,I)$ problem representing 
$\log n$ copies of each of the $n$ instances. The players will compute all the copies in parallel 
with error $<2\ve$, and then take a majority of the $\log n$ copies for each instance. For a sufficiently large
$n$ this guarantees the correct answer for all $n$ instances except with probability $<\ve$. Thus our goal is 
to simulate $m$ copies of $\CPJ(C,I)$.
We view $\CPJ(C,I)^m$ as a degree-$2^m$, $C$-round correlated pointer jumping problem in the natural way. 
Each node represents a vector $V=(v_1,\ldots,v_m)$ of $m$ nodes in the $m$ copies of  $\CPJ(C,I)$. The children of $V$ are the 
$2^m$ possible combinations of children of $\{v_1,\ldots,v_m\}$. The distribution on the children is the product distribution 
induced by the distributions in $v_1,\ldots,v_m$. We claim that 
\begin{equation}
\label{eq:divCPJ}
\DivProt{\CPJ(C,I)^n_{v_1,\ldots,v_m}} = \sum_{i=1}^m \DivProt{\CPJ(C,I)_{v_i}}.
\end{equation}
This follows easily by induction on the tree, since the distribution on each node
is a product distribution, and for each independent pairs $(P_1,Q_1),\ldots,(P_m,Q_m)$ we 
have
$$
\Div{P_1\times P_2\times\ldots\times P_m}{Q_1\times Q_2\times\ldots\times Q_m} = 
\Div{P_1}{Q_1}+\ldots+\Div{P_m}{Q_m},
$$
by \lemmaref{lemma:divproduct}.
By applying \eqref{eq:divCPJ} to the root of the tree we see that $\DivProt{\CPJ(C,I)^m} \le m\cdot I\le C\log n$. Thus 
Theorem~\ref{theorem:pointer} implies that ${\CPJ(C,I)^n}$ can be solved with an additional error of $\ve/2$ using an expected
$$
 C\log n + C \log C/\ve + o(C\log n)
$$
bits of communication.
\end{proof}

\section{Acknowledgments}
We thank Boaz Barak and Xi Chen for useful discussions.

\bibliographystyle{alphasy1}
%Gather{refs.bib}
%input "refs.bib"
\bibliography{refs}

\end{document}

%% file: macros.tex
\usepackage{nicefrac}
\usepackage{xspace}

\newcommand{\getsr}{\in_{_{\!\text{\fontsize{4}{4}\selectfont
R}}}}

\newcommand{\e}{\epsilon}

\newcommand{\ring}[1]{\mathbb{#1}}

\newcommand{\Z}{\ring{Z}}

\newcommand{\bits}{\{0,1\}}